\documentclass[11pt]{article}   
\usepackage{amsthm,amsmath}    
\usepackage[dvips]{color}
\usepackage{amssymb}
\usepackage{latexsym}
\usepackage{hyperref}

\theoremstyle{plain}
\newtheorem{theorem}{Theorem}[section]
\newtheorem{lemma}[theorem]{Lemma}
\newtheorem{corollary}[theorem]{Corollary}

\theoremstyle{definition}
\newtheorem{definition}[theorem]{Definition}
\newtheorem{ex}[theorem]{Example}
\newtheorem{step}{Step}

\def\keywords{%
\list{}{\advance\topsep by0.35cm\relax\small\rm
 \leftmargin=1cm
 \itemindent\listparindent
 \rightmargin\leftmargin}\item[\hskip\labelsep\bf Keywords: ]}

\theoremstyle{remark}

\renewenvironment{proof}{\noindent{\it Proof}.}{\qed}

\title{On Some Complexity Results for Even Linear Languages}  
 
\author{Liliana Cojocaru \\
Department of Computer Science\\ 
"Alexandru Ioan Cuza" University 
   and\\
"Petre Andrei" University of Ia\c{s}i \\
lcojocarus{@}yahoo.com
}
\date{}

\begin{document}
\thispagestyle{empty}
\maketitle
\thispagestyle{empty}
\pagenumbering{arabic}

\begin{abstract}
We deal with a normal form for \textit{context-free grammars}, called \textit{Dyck normal form}. 
This normal form is a syntactical restriction of the \textit{Chomsky normal form}, in which the 
two nonterminals occurring on the right-hand side of a rule are paired nonterminals. This pairwise 
property, along with several other terminal rewriting conditions, makes it possible to define a 
homomorphism from Dyck words to words generated by a grammar in Dyck normal form. We prove that 
for each context-free language $L$, there exist an integer $K$ and a homomorphism $\varphi$ such 
that $L=\varphi(D'_K)$, where $D'_K\subseteq D_K$ and $D_K$ is the one-sided Dyck language over 
$K$ letters. As an application we give an alternative proof of the inclusion of the class of 
\textit{even linear languages} in $\cal A$$\cal C$$^1$. 

\end{abstract}
\vspace*{-0.7cm}

\section{Introduction}
A normal form for context-free grammars (CFGs) consists of restrictions 
imposed to the structure of context-free (CF) productions, especially on the number 
of terminals and nonterminals allowed on the right-hand side of a CF rule. Normal 
forms have turned out to be useful tools in studying syntactical properties 
of CFGs, in parsing theory, structural and descriptional complexity, inference and  
learning theory. Various normal forms for CFGs have been studied, but the most 
important still remain the \textit{Chomsky normal form}, \textit{Greibach normal form}, 
and \textit{operator normal form}. For definitions, results, and surveys on normal 
forms the reader is referred to \cite{H} and  \cite{PL}. 

A normal form is correct if it preserves the language generated by the original grammar. 
This condition is called {\it the weak equivalence}, i.e., a normal form preserves the 
language but may lose important syntactical or semantical properties of the original grammar. 
It is well known that Chomsky and Greibach normal forms are ``strong'' normal form in this 
respect. They preserve almost all syntactical and semantical properties of the original grammar. 

This paper is devoted to a new normal form for CFGs, called \textit{Dyck normal form}. The Dyck 
normal form is a syntactical restriction of the Chomsky normal form, in which the 
two nonterminals occurring on the right-hand side of a rule are paired nonterminals, in the 
sense that each left (right) nonterminal of a pair has a unique right (left) pairwise. This 
pairwise property imposed on the structure of the right-hand side of each rule induces a nested 
structure on the derivation tree of each word generated by a grammar in Dyck normal form. More 
precisely, each derivation tree of a word generated by a grammar in Dyck normal form, read in 
the depth-first order is a Dyck word, hence the name of the normal form. We have been 
inspired to develop this new normal form by the general theory of Dyck words and Dyck languages, 
that turned out to play a crucial role in the description and characterization of CFLs 
\cite{EHR2}, \cite{EHR3}, and \cite{IJR}. 

Furthermore, there exists always a homomorphism laying between the derivation tree of a word 
generated by a grammar in Chomsky normal form and its equivalent in Dyck normal form. In other 
words, Dyck and Chomsky normal forms are strongly equivalent. This property, along with several 
other terminal rewriting conditions imposed to a grammar in Dyck normal form, enable us to 
define a homomorphism from Dyck words to words generated by a grammar in Dyck normal form.

The definition and several properties of grammars in Dyck normal form are presented in Section 2. 
By exploiting the Dyck normal form, and several characterizations of Dyck languages presented in 
\cite{IJR}, we give a new characterization of CFLs in terms of Dyck languages. We prove in Section 
3, that for each CF language $L$ there exist an integer $K$ and a homomorphism $\varphi$ such that 
$L=\varphi(D'_K)$, where $D'_K$ is a subset of the one-sided Dyck language over $K$ letters. 
Using this homomorphism we prove in Section 4 that \textit{even linear languages} \cite{AP}, \cite{KM}, \cite{SG} can be 
accepted by alternating Turing machines (ATMs) in logarithmic space, square logarithmic time, and 
a logarithmic number of alternations, which is an alternative proof of the inclusion of the class 
of even linear languages in $U_{L}$-uniform $\cal A$$\cal C$$^1$. 

\section{Dyck Normal Form}

We assume the reader to be familiar with the basic notions of formal language theory. For an alphabet $X$, $X^*$ 
denotes the free monoid generated by $X$. By $|x|_a$ we denote the number of occurrences of the letter $a$ in the 
string $x\in X^*$, while $|x|$ is the length of $x\in V^*$. We denote by $\lambda$ the empty string. If $X$ is a 
finite set, then $|X|$ is the cardinality of $X$. 

\begin{definition}
A CFG\footnote{\textit{A context-free grammar} is denoted by $G=(N, T, P, S)$, 
where $N$ and $T$ are finite sets of \textit{nonterminals} and \textit{terminals}, respectively, $N\cap T= \emptyset$,
$S\in N-T$ is the grammar \textit{axiom}, and $P\subset N\times (N\cup T)^*$ is the finite set of productions.} 
$G=(N, T, P, S)$ is said to be in {\it Dyck normal form} if it satisfies the following conditions:
\begin{enumerate}
  \item G is in Chomsky normal form, 
  \item if $A\rightarrow a \in P$, $A\in N$, $A\neq S$, $a\in T$, 
   then no other rule in $P$ rewrites $A$, 
  \item for each $A\in N$ such that $X\rightarrow AB \in P$
   ($X\rightarrow BA \in P$) there is no other rule in $P$ of
  the form $X'\rightarrow B'A$ ($X'\rightarrow AB'$),
  \item for each rules $X\rightarrow AB$, $X'\rightarrow
  A'B$ ($X\rightarrow AB$, $X'\rightarrow AB'$), we have $A=A'$ ($B=B'$).
\end{enumerate}
\end{definition}

Note that the Dyck normal form  is a Chomsky normal form  on which several restrictions 
on the positions of nonterminals occurring on the right hand-side of each context-free 
production are imposed. The reasons for which we introduce these restrictions (items 2, 
3, and 4 in Definition 2.1) are the following. The second condition in Definition 2.1 
allows to make a partition between those nonterminals rewritten by nonterminals, and 
those nonterminals rewritten by terminals. This enables, in Section 3, to define 
a homomorphism from Dyck words to words generated by a grammar in Dyck normal form. 
The third and fourth conditions allow us to split the set of nonterminals into pairwise 
nonterminals, and thus to introduce bracketed pairs. Next we prove that the Dyck normal form  is correct. 

\begin{theorem}
For each CFG  $G=(N, T, P, S)$ there exists a
grammar $G'=(N', T, P', S)$ such that $L(G)=L(G')$ and $G'$ is in
Dyck normal form.
\end{theorem}

\begin{proof}
Suppose that $G$ is a CFG in Chomsky normal form. Otherwise, using the algorithm as 
described in \cite{H} or \cite{PL} we can convert $G$ into Chomsky normal form. 
In order to convert $G$ from Chomsky normal form into Dyck normal form  we proceed as follows.

\begin{step}\rm We check  whether $P$ contains two (or more) 
rules of the form $A\rightarrow a$, $A\rightarrow b$, $a\neq b$.  
If it does, then for each rule $A\rightarrow b$, $a\neq b$, we
introduce a new variable $A_b$. We add the new rule $A_b\rightarrow b$,
and we remove the rule $A\rightarrow b$. For each rule of the form
$X\rightarrow AB$ ($X\rightarrow BA$) we add the new rule
$X\rightarrow A_bB$ ($X\rightarrow BA_b$), while for a rule of
the form $X\rightarrow AA$ we add three new rules $X\rightarrow
A_bA$, $X\rightarrow AA_b$, $X\rightarrow A_bA_b$, without removing
the initial rules. We call this procedure an $A_b$-{\it terminal 
substitution} of $A$.

For each rule $A\rightarrow a$, $a\in T$, we check whether a rule of the form 
$A\rightarrow B_1B_2$, $B_1,B_2\in N$, exists in $P$. If it does, then a new  
nonterminal $A_a$ is introduced and we perform an $A_a$-terminal substitution 
of  $A$ for the rule $A\rightarrow a$.
\end{step}

\begin{step}\rm Suppose that there exist two (or more) rules of the form 
$X\rightarrow AB$ and $X'\rightarrow B'A$. If we have agreed on preserving 
only the left occurrences of $A$ in the right-hand sides, then according to 
condition 3 of Definition 2.1, we have to remove all right occurrences of $A$. 
To do this we introduce a new nonterminal ${}_ZA$ and all right occurrences 
of $A$, preceded at the left side by $Z$, in the right-hand side of a rule, 
are substituted by ${}_ZA$. For each rule that rewrites $A$, $A\rightarrow Y$, 
$Y\in N^2\cup T$, we add a new rule of the form ${}_ZA\rightarrow Y$, preserving 
the rule $A\rightarrow Y$. We call this procedure an ${}_ZA$-{\it nonterminal 
substitution} of $A$. 

According to this procedure, for the rule $X'\rightarrow B'A$, 
we introduce a new nonterminal ${}_{B'}A$, we add the rule $X'\rightarrow B'{}_{B'}A$, 
and remove the rule $X'\rightarrow B'A$. For each rule that rewrites $A$, of the 
form\footnote{This case deals with the possibility of having $Y=B'{}_{B'}A$, 
too.} $A\rightarrow Y$, $Y\in N^2\cup T$, we add a new rule of the form 
${}_{B'}A\rightarrow Y$, preserving the rule $A\rightarrow Y$. 
\end{step}

\begin{step}\rm  Finally, for each two rules $X\rightarrow AB$, $X'\rightarrow
A'B$ ($X\rightarrow BA$, $X'\rightarrow BA'$) with $A\neq A'$, a new nonterminal 
${}_{A'}B$ ($B_{A'}$) is introduced to replace $B$ from the second rule, and we  
perform an ${}_{A'}B$($B_{A'}$)-nonterminal substitution of $B$, i.e., we add 
$X'\rightarrow A'{}_{A'}B$, and remove $X'\rightarrow A'B$. For each rule that 
rewrites $B$, of the form $B\rightarrow Y$, $Y\in N^2\cup T$, we add a new rule 
${}_{A'}B\rightarrow Y$ by preserving the rule $B\rightarrow Y$. 

In the case that $A'$ occurs on the right-hand side of another rule, such
that $A'$ matches at the right side with another nonterminal
different of ${}_{A'}B$, then the procedure described above is repeated
for $A'$, too.
\end{step}

Note that, if one of the conditions 2, 3, and 4 in Definition 2.1, has been settled, we 
do not have to resolve it once again in further steps of this procedure. 

The new grammar 
$G'=(N', T, P', S)$  built in steps $1-3$ has all nonterminals from $N$ and productions 
from $P$, plus/minus all nonterminals and productions, respectively introduced/removed 
according to the substitutions performed during the above steps. 

Next we prove that grammars $G=(N, T, P, S)$, in Chomsky normal form, and $G'=(N', T, P', S)$, 
in Dyck normal form, generate the same language. In this order, consider the homomorphism 
$h_d: N'\cup T \rightarrow N \cup T$ defined by $h_d(x)=x$, $x\in T$, $h_d(X)=X$, 
for $X\in N$, and $h_d(X')=X$ for $X'\in N'-N$, $X\in N$ such that $X'$ is a 
(transitive\footnote{There exist $X_k\hspace{-0.1cm}\in \hspace{-0.1cm}N$, such that $X'$ 
is an $X'$-substitution of $X_k$, $X_k$ is an $X_k$-substitution of $X_{k-1}$,..., and 
$X_1$ is an $X_1$-substitution of $X$. All of them substitute $X$.}) $X'$-substitution of 
$X$, terminal or not, in the above construction of the grammar $G'$.

To prove that $L(G') \subseteq L(G)$ we extend $h_d$ to a
homomorphism from $(N'\cup T)^*$ to $(N \cup T)^*$ defined on the
classical concatenation operation. It is straightforward to prove by
induction, that for each $\alpha \Rightarrow^*_{G'} \delta$ we have
$h_d(\alpha) \Rightarrow^*_G h_d(\delta)$. This implies that for any 
derivation of a word $w\in L(G')$, i.e., $S\Rightarrow^*_{G'}w$, we 
have $h_d(S) \Rightarrow^*_G h_d(w)$, i.e., $S \Rightarrow^*_G w$, 
or equivalently, $L(G') \subseteq L(G)$. 

To prove that $L(G) \subseteq L(G')$ we make use of the CYK
(Cocke-Younger-Kasami) algorithm as described in \cite{PL}.
  
Let $w=a_1a_2...a_n$ be an arbitrary word from $L(G)$, and $V_{ij}$, $i\leq j$, 
$i,j\in \{1,..., n\}$, be the triangular matrix of size $n\times n$ built with the 
CYK algorithm \cite{PL}. Since $w\in L(G)$, we have $S \in V_{1n}$. We prove 
that $w\in L(G')$, i.e., $S \in V'_{1n}$, where $V'_{ij}$, $i\leq j$, $i,j\in \{1,..., n\}$ 
forms the triangular matrix obtained by applying the CYK algorithm to $w$ according to the 
$G'$ productions. 

We consider two relations $\hat h_t: N\cup T \rightarrow N'\cup T$ and
$\hat h_{\neg t}: N \rightarrow N'$. The first relation is defined by
$\hat h_t(x)=x$, $x\in T$, $\hat h_t(S)=S$, if $S\rightarrow x$, $x\in T$,
is a rule in $G$, and $\hat h_t(X)=X'$, if $X'$ is a (transitive) $X'$-terminal
substitution\footnote{There may exist several (distinct) terminal (nonterminal) 
substitutions for the same nonterminal $X$. This property makes $\hat h_t$ ($\hat h_{\neg t}$) 
to be a relation.} of $X$, and $X\rightarrow x$ is a rule in $G$. Finally, $\hat h_t(X)=X$ if 
$X\rightarrow x \in P$, $x\in T$. The second relation is defined as 
$\hat h_{\neg t}(S)=S$, $\hat h_{\neg t}(X)=\{X\}\cup \{ X'| X' \mbox{ is a
(transitive) } X'\mbox{-nonterminal substitution of } X\}$ and $\hat h_{\neg
t}(X)=X$, if there is no substitution of $X$ and no rule of the form $X\rightarrow x$, 
$x\in T$, in $G$. 

Notice that $\hat h_y(X_1\cup X_2)=\hat h_y(X_1)\cup \hat h_y(X_2)$, 
for $X_i \subseteq N$, $i\in \{1,2\}$, $y\in \{t, \neg t\}$. Furthermore,
using the relation $\hat h_t$, each rule $X\rightarrow x$ in $P$ 
has a corresponding set of rules $\{X'\rightarrow x| X'\in \hat h_t(X),
X\rightarrow x \in P\}$ in $P'$. Each rule $A\rightarrow BC$ in $P$ has a 
corresponding set of rules  $\{A'\rightarrow B'C'| A'\in \hat h_{\neg t}(A),
B'\in \hat h_{\neg t}(B)\cup \hat h_t(B), C'\in \hat h_{\neg t}(C)\cup\hat h_t(C),
B' \mbox{ and } C' \mbox{ are pairwise nonterminals, } A\rightarrow BC \in P\}$ in $P'$.

Consider $V'_{ii}= \hat h_t(V_{ii})$ and $V'_{ij}= \hat
h_{\neg t}(V_{ij})$, $i < j$, $i,j\in \{1,..., n\}$. We claim
that $V'_{ij}$, $i,j\in \{1,..., n\}$, $i\leq j$, defined as before, 
forms the triangular matrix obtained by applying the CYK algorithm 
to rules that derive $w$ in $G'$.

First, observe that for $i=j$, we have 
$V'_{ii}=\hat h_t(V_{ii})=\{A | A\rightarrow a_i\in P'\}$, $i\in \{1,..., n\}$, due to 
the definition of the relation $\hat h_t$. Now consider $k=j-i$, $k\in \{1,..., n-1\}$. 
We want to compute $V'_{ij}$, $i<j$.

By definition, we have $V_{ij}=\bigcup_{l=i}^{j-1} \{A| A\rightarrow BC, B\in V_{il}, 
C\in V_{l+1j}\}$, so that $V'_{ij}=\hat h_{\neg t}(V_{ij})\hspace{-0.1cm}=\hat h_{\neg t}(\bigcup_{l=i}^{j-1} \{A|
A\rightarrow BC, B\in V_{il}, C\in V_{l+1j}\})\hspace{-0.1cm}=\bigcup_{l=i}^{j-1}
\hat h_{\neg t}(\{A| A\rightarrow BC, B\in V_{il}, C\in V_{l+1j}\})=
\bigcup_{l=i}^{j-1}\{A'| A'\rightarrow B'C', A'\in \hat h_{\neg t}(A),
B'\in \hat h_{\neg t}(B)\cup \hat h_t(B), B\in V_{il},
C'\in \hat h_{\neg t}(C)\cup\hat h_t(C), C\in V_{l+1j}$,
$B'$ and $C'$ are pairwise nonterminals, $A\rightarrow BC \in P\}$.
Let us explicitly develop the last union.

If $k=1$, then $l\in \{i\}$. For each $i\in \{1,..., n-1\}$ we have
$V'_{ii+1}=\{A'| A'\rightarrow B'C', A'\in \hat h_{\neg t}(A), B'\in
\hat h_{\neg t}(B)\cup \hat h_t(B), B\in V_{ii}, C'\in \hat h_{\neg
t}(C)\cup\hat h_t(C), C\in V_{i+1i+1}$, $B'$ and $C'$ are pairwise
nonterminals, $A\rightarrow BC \in P\}$. Since $B\in
V_{ii}$ and $C\in V_{i+1i+1}$, $B'$ is a terminal substitution of
$B$, while $C'$ is a terminal substitution of $C$. Therefore, we have
$B'\notin \hat h_{\neg t}(B)$, $C'\notin\hat h_{\neg t}(C)$, so that
$B'\in \hat h_t(B)$, for all $B\in V_{ii}$, and $C'\in \hat h_t(C)$, 
for all $C\in V_{i+1i+1}$, i.e., $B'\in \hat h_t(V_{ii})=V'_{ii}$
and $C'\in \hat h_t(V_{i+1i+1})=V'_{i+1i+1}$. Therefore,
$V'_{ii+1}=\{A'| A'\rightarrow B'C', B'\in V'_{ii}, C'\in
V'_{i+1i+1}\}$.

If $k\geq 2$, then $l\in \{i, i+1,..., j-1\}$, and
$V'_{ij}=\bigcup_{l=i}^{j-1}\{A'| A'\rightarrow B'C', A'\in \hat h_{\neg t}(A),
B'\in \hat h_{\neg t}(B)\cup \hat h_t(B), B\in V_{il},
C'\in \hat h_{\neg t}(C)\cup\hat h_t(C), C\in V_{l+1j}$,
$B'$ and $C'$ are pairwise nonterminals, $A\rightarrow BC \in P\}$.

We now compute the first set of the above union, i.e., $V'_i=\{A'|
A'\rightarrow B'C', A'\in \hat h_{\neg t}(A), B'\in \hat h_{\neg
t}(B)\cup \hat h_t(B), B\in V_{ii}, C'\in \hat h_{\neg t}(C)\cup\hat
h_t(C), C\in V_{i+1j}$, $B'$ and $C'$ are pairwise nonterminals,
$A\rightarrow BC \in P\}$. By the same reasoning as before, the
condition $B'\in \hat h_{\neg t}(B)\cup \hat h_t(B), B\in V_{ii}$,
is equivalent with $B'\in \hat h_t(V_{ii})=V'_{ii}$.

Because $i+1 \neq j$, $C'$ is a nonterminal substitution of $C$. 
Therefore, $C'\notin \hat h_t(C)$, and the condition
$C'\in \hat h_{\neg t}(C)\cup \hat h_t(C), C\in V_{i+1j}$ is equivalent
with $C'\in \hat h_{\neg t}(V_{i+1j})=V'_{i+1j}$. So that
$V'_i= \{A'| A'\rightarrow B'C', B'\in V'_{ii}, C'\in V'_{i+1j}\}$.

Using the same method for each $l\in \{i+1,..., j-1\}$ we have
 $V'_l=\{A'| A'\rightarrow B'C', A'\in \hat h_{\neg t}(A), B'\in \hat
h_{\neg t}(B)\cup \hat h_t(B), B\in V_{il}, C'\in \hat h_{\neg
t}(C)\cup\hat h_t(C), C\in V_{l+1j}$, $B'$ and $C'$ are pairwise
nonterminals, $A\rightarrow BC \in P\}= \{A'| A'\rightarrow B'C',
B'\in V'_{il}, C'\in V'_{l+1j}\}$. 

In conclusion, $V'_{ij}=\bigcup_{l=i}^{j-1} \{A'| A'\rightarrow B'C', 
B'\in V'_{il}, C'\in V'_{l+1j}\}$, for each $i,j\in \{1,..., n\}$, i.e., 
$V'_{ij}$, $i\leq j$, contains the nonterminals of the $n\times n$ 
triangular matrix computed by applying the CYK algorithm to rules that 
derive $w$ in $G'$. Because $w\in L(G)$, we have $S\in V_{1n}$. That is 
equivalent with $S\in V'_{1n}=\hat h_t(V_{1n})$, if $n=1$, and 
$S\in V'_{1n}=\hat h_{\neg t}(V_{1n})$, if $n>1$, i.e., $w\in L(G')$.
\end{proof}

\begin{corollary}
Let $G$ be a CFG in Dyck normal form. Any terminal derivation  
in $G$ producing a word of length $n$, $n\geq 1$, takes $2n-1$ steps.
\end{corollary}

\begin{proof}
If $G$ is a CFG in Dyck normal form, then it is also in Chomsky normal form, and all properties of the latter hold.
\end{proof}

\begin{corollary}
If $G=(N, T, P, S)$ is a grammar in Chomsky normal form, and $G'=(N', T, P', S)$ its equivalent in 
Dyck normal form, then there exists a homomorphism $h_d:N'\cup T \rightarrow N \cup T$, such
that any derivation tree of $w\in L(G)$ is the homomorphic image of
a derivation tree of the same word in $G'$.
\end{corollary}

\begin{proof}
Consider the homomorphism $h_d:N'\cup T \rightarrow N \cup T$  
defined as $h_d(S)=S$, $h_d(A_t)= h_d({}_ZA)= h_d(A_Z)=A$, for 
each $A_t$-terminal or ${}_ZA$($A_Z$)-nonterminal substitution 
of $A$, $A_t, {}_ZA, A_Z \in N'-N$, $h_d(A)= A$,  $A\in N$, and 
$h_d(t)=t$, $t\in T$. The claim is a direct consequence of the 
way in which the new nonterminals $A_t$, ${}_ZA$, and $A_Z$ have 
been considered in the proof of Theorem 2.2. 
\end{proof}

Let $G$ be a grammar in Dyck normal form. In order to emphasis the pairwise 
brackets occurring on the right-hand side of each rule,  each nonterminal 
$A$ and $B$ occurring on a rule of the form $X\rightarrow AB$, is 
replaced by a left and right bracket $[_A$ and $]_B$, respectively. In each 
rule that rewrites $A$ or $B$, we replace $A$ by $[_A$, and $B$ by $]_B$, 
respectively. 

Next, we present an example of the conversion procedure described in Theorem 2.2, 
along with the homomorphism that lays between the grammars in Chomsky normal form 
and its equivalent in Dyck normal form.

\begin{ex}\rm 
Consider the CFG  $G=( \{E, T, R\}, \{ +,*,a\}, E, P)$ with  
$P=\{ E\rightarrow a/T*R/E+T, T\rightarrow a/T*R, R\rightarrow a \}$.

The Chomsky normal form of $G$ is $G'=(\{E_0,E,E_1,E_2,T,T_1,T_2,R\}, 
\{+,*, a\}, E_0, P')$ in which $P'=\{E_0\rightarrow a/TT_1/EE_1,
E\rightarrow a/TT_1/EE_1, T\rightarrow a/TT_1, T_1\rightarrow T_2R, 
E_1\rightarrow E_2T, T_2\rightarrow *, E_2\rightarrow +, R\rightarrow a\}$.

We now convert $G'$ into Dyck normal form. To do this, with
respect to Definition 2.1, item 2, we first have to remove 
$E\rightarrow a$ and $T\rightarrow a$. Then, according to
item 3, we remove the right occurrence of $T$ from the rule
$E_1\rightarrow E_2T$, along with other transformations that may be
required after completing these procedures. Let $E_3$ and $T_3$ be two 
new nonterminals. We remove  $E\rightarrow a$ and $T\rightarrow a$, and 
add the rules $E_3\rightarrow a$, $T_3\rightarrow a$, $E_0\rightarrow E_3E_1$, 
$E_0\rightarrow T_3T_1$, $E\rightarrow E_3E_1$, $E\rightarrow T_3T_1$,
$E_1\rightarrow E_2T_3$, and $T\rightarrow T_3T_1$.

Let $T'$ be the new nonterminal that replaces the right occurrence
of $T$. We add the rules $E_1\rightarrow E_2T'$, $T'\rightarrow TT_1$, and 
$T'\rightarrow T_3T_1$, and remove $E_1\rightarrow E_2T$.
We repeat the procedure with $T_3$ (added in the previous step), i.e.,
we introduce a new nonterminal $T_4$, remove $E_1\rightarrow E_2T_3$,
add $E_1\rightarrow E_2T_4$ and  $T_4 \rightarrow a$. Due to the new 
nonterminals $E_3$, $T_3$, and $T_4$, item 4  does not hold. 
To accomplish this condition too, we introduce three new nonterminals 
$E_4$ to replace $E_2$ in $E_1\rightarrow E_2T_4$, $E_5$ to replace $E_1$ 
in $E_0\rightarrow E_3E_1$ and $E\rightarrow E_3E_1$, and $T_5$ to
replace $T_1$ in  $E_0\rightarrow T_3T_1$ and $E\rightarrow T_3T_1$.
We remove all the above rules and add the new rules $E_1\rightarrow
E_4T_4$, $E_4 \rightarrow +$, $E_0\rightarrow E_3E_5$, $E\rightarrow
E_3E_5$, $E_5\rightarrow E_2T'$, $E_5\rightarrow E_4T_4$,
$E_0\rightarrow T_3T_5$, $E\rightarrow T_3T_5$, and $T_5\rightarrow
T_2R$. 

The Dyck normal form  of $G'$, using the bracket notation, is 
$G''=( \{E_0,[_E,]_{E_1},[_{E_2},[_{E_3}, 
[_{E_4}, $\\$]_{E_5}, ]_{T'},[_T,]_{T_1},[_{T_2},[_{T_3}, ]_{T_4},]_{T_5},]_R\}, 
\{+,*, a\}, E_0,P'')$, where 

$P''=\{E_0\rightarrow a/[_T\hspace{0.1cm}
]_{T_1}/[_E\hspace{0.1cm}]_{E_1}/[_{E_3}\hspace{0.1cm}]_{E_5}/[_{T_3}\hspace{0.1cm}]_{T_5},[_E\hspace{0.1cm}\rightarrow [_T\hspace{0.1cm}]_{T_1}/
[_E\hspace{0.1cm}]_{E_1}/[_{E_3}\hspace{0.1cm}]_{E_5}/[_{T_3}\hspace{0.1cm}]_{T_5},]_{E_1}\rightarrow[_{E_2}\hspace{0.1cm}]_{T'}/\hspace{0.1cm}[_{E_4}\hspace{0.1cm}]_{T_4}, ]_{E_5}\rightarrow[_{E_2}\hspace{0.1cm}]_{T'}/[_{E_4}\hspace{0.1cm}]_{T_4}, \hspace{0.1cm}[_T\hspace{0.1cm}\rightarrow\hspace{0.1cm}[_T\hspace{0.1cm}]_{T_1}/[_{T_3}\hspace{0.1cm}]_{T_5},
\hspace{0.1cm}]_{T'}\rightarrow[_T\hspace{0.1cm}]_{T_1}/\hspace{0.1cm}[_{T_3}\hspace{0.1cm}]_{T_5}, 
\hspace{0.1cm}]_{T_1}\rightarrow \hspace{-0.1cm}[_{T_2}\hspace{0.1cm}]_R,\hspace{0.1cm} $\\$]_{T_5}\rightarrow[_{T_2}\hspace{0.1cm}]_R, 
[_{T_2}\rightarrow *,\hspace{0.1cm}[_{T_3}\rightarrow a,\hspace{0.1cm}]_{T_4}\rightarrow a,\hspace{0.1cm}[_{E_2}\rightarrow
+,\hspace{0.1cm} [_{E_3}\rightarrow a,\hspace{0.1cm}[_{E_4}\rightarrow +,\hspace{0.1cm}]_R\rightarrow a\}$.

The homomorphism $h_{d}$ is defined as $h_d:N'\cup T \rightarrow N'' \cup T$,  
$h_d(E_0)= E_0$, $h_d([_E)= h_d([_{E_3})=E$, $h_d(]_{E_1})= h_d(]_{E_5})=E_1$,
$h_d([_{E_2})= h_d([_{E_4})=E_2$, $h_d([_T)= h_d(]_{T'})=h_d([_{T_3})=h_d(]_{T_4})=T$,
$h_d(]_{T_1})= h_d(]_{T_5})=T_1$, $h_d([_{T_2})=T_2$, $h_d(]_R)=R$,
$h_d(t)= t$, for each $t\in T$. 

The string $w=a*a*a+a$ is a word in $L(G'')=L(G)$ generated by a 
leftmost derivation $D$ in $G''$ as follows:

$D:$ $E_0\Rightarrow [_E\hspace{0.1cm}]_{E_1}\Rightarrow
[_T\hspace{0.1cm}]_{T_1}\hspace{0.1cm}]_{E_1}\Rightarrow
[_{T_3}\hspace{0.1cm}]_{T_5}\hspace{0.1cm}]_{T_1}\hspace{0.1cm}
]_{E_1}\Rightarrow a\hspace{0.1cm}]_{T_5}\hspace{0.1cm}]_{T_1}\hspace{0.1cm}
]_{E_1}\Rightarrow a\hspace{0.1cm}[_{T_2}\hspace{0.1cm}]_R\hspace{0.1cm}]_{T_1}\hspace{0.1cm}
]_{E_1}\Rightarrow a\hspace{0.1cm}*\hspace{0.1cm}]_R\hspace{0.1cm}]_{T_1}\hspace{0.1cm}
]_{E_1}\Rightarrow a\hspace{0.1cm}*\hspace{0.1cm}a\hspace{0.1cm}]_{T_1}\hspace{0.1cm}
]_{E_1}\Rightarrow a\hspace{0.1cm}*\hspace{0.1cm}a\hspace{0.1cm}[_{T_2}\hspace{0.1cm}]_R
\hspace{0.1cm} ]_{E_1}\Rightarrow a\hspace{0.1cm}*\hspace{0.1cm}a\hspace{0.1cm}*\hspace{0.1cm}]_R
\hspace{0.1cm} ]_{E_1}\Rightarrow a\hspace{0.1cm}*\hspace{0.1cm}a\hspace{0.1cm}*\hspace{0.1cm}a
\hspace{0.1cm} ]_{E_1}\Rightarrow a\hspace{0.1cm}*\hspace{0.1cm}a\hspace{0.1cm}*\hspace{0.1cm}a
\hspace{0.1cm} [_{E_4}\hspace{0.1cm}]_{T_4}\Rightarrow
a\hspace{0.1cm}*\hspace{0.1cm}a\hspace{0.1cm}*\hspace{0.1cm}a
\hspace{0.1cm} +\hspace{0.1cm}]_{T_4}\Rightarrow 
a\hspace{0.1cm}*\hspace{0.1cm}a\hspace{0.1cm}*\hspace{0.1cm}a
\hspace{0.1cm} +\hspace{0.1cm}a$.

Applying $h_{d}$ to the derivation $D$ of $w$ in $G''$ we obtain a derivation 
of $w$ in $G'$. If we consider the derivation tree $\cal T$ of $w$ in $G$, and   
the derivation tree $\cal T'$ of $w$ in $G''$, then $\cal T$ is the homomorphic 
image of $\cal T'$ through $h_d$. 
\end{ex}

\section{Characterizations of CFLs by Dyck Languages}

\begin{definition}\rm
Let $G_k=(N_k, T, P_k, S)$ be a CFG in Dyck normal form  with $|N_k-\{S\}|=2k$. 
Let $D: $ $S\Rightarrow u_1 \Rightarrow u_2\Rightarrow
\hspace*{-0.1cm}...\hspace*{-0.1cm} \Rightarrow u_n=w$, $n\geq 3$, be a leftmost
derivation of $w\in L(G)$. The {\it trace-word} of $w$ associated with the 
derivation $D$, denoted as $t_{w,D}$, is defined as the concatenation of 
nonterminals consecutively rewritten in $D$, excluding the axiom. The 
{\it trace-language} associated with $G_k$, denoted as $\L_k$, is  
$\L_k=\{t_{w,D}| D\mbox{ is a leftmost derivation of } w, \\w\in L(G_k) \}$.
\end{definition}

Note that the trace-word $t_{w,D}$ associated with a word $w\in L(G)$  
is defined only for those derivations $D$, in a CFG in Dyck normal form, of length 
at least $3$. Derivations of length $2$ do not exist in a grammar in Dyck normal 
form, while derivations of length $1$ correspond to rules of the form 
$S\rightarrow t$, $t\in T$, hence no nonterminal interferes during the derivation, 
unless the axiom. $t_{w,D}$, $w\in L(G)$, can also be read from the derivation tree 
in a depth-first order starting with the root, but ignoring the root (the axiom) 
and the leaves (terminals).  

The trace-word associated to $w=a*a*a+a$ and the leftmost derivation $D$, in 
Example 2.5, is $t_{a*a*a+a,D}=[_E\hspace{0.1cm}[_T\hspace{0.1cm}[_{T_3}
\hspace{0.1cm}]_{T_5}\hspace{0.1cm}[_{T_2}\hspace{0.1cm}]_R
\hspace{0.1cm}]_{T_1}\hspace{0.1cm}[_{T_2}\hspace{0.1cm}]_R
\hspace{0.1cm}]_{E_1}\hspace{0.1cm}[_{E_4}\hspace{0.1cm}]_{T_4}$.

\begin{definition}\rm
A one-sided Dyck language over $k$ letters, $k\geq 1$, is a
context-free language defined by the grammar $\Gamma_k=(\{S\}, T_k, P, S)$,
where $T_k=\{[_1,\hspace{0.1cm}[_2,...,[_k,
\hspace{0.1cm}]_1,\hspace{0.1cm}]_2,...,\hspace{0.1cm}]_k \}$ and
$P=\{S\rightarrow [_i\hspace{0.1cm}S\hspace{0.1cm}]_i, S\rightarrow
SS,S\rightarrow [_i\hspace{0.1cm}]_i \hspace{0.1cm}|\hspace{0.1cm}
1\leq i \leq k\}.$
\end{definition}

Let $G_k=(N_k, T, P_k, S)$ be a CFG in Dyck normal form. In order to keep 
control of each bracketed pair occurring in the right-hand side of each 
rule from $G_k$, we fix $N_k=\{S, [_1,\hspace{0.1cm}
[_2,...,[_k,\hspace{0.1cm}]_1,\hspace{0.1cm}]_2,...,\hspace{0.1cm}]_k \}$, 
and $P_k$ to be composed of rules of the form $X \rightarrow
[_i\hspace{0.1cm}]_i$, $1\leq i \leq k$, and $Y \rightarrow t$, $X, Y
\in N_k$,  $t \in T$. 

From \cite{IJR} we have adopted the next characterizations 
of $D_k$, $k\geq 1$, (Definition 3.3, and Lemmas 3.4 and 3.5). 

\begin{definition}\rm
For a string  $w$, let $w_{i:j}$ be its substring starting at the
$i^{th}$ position and ending at the $j^{th}$ position. Let $h$ be a
homomorphism defined as follows $h([_1)=h([_2)=...=h([_k)=[_1,$
$h(]_1)=h(]_2)=...=h(]_k)=]_1$. Let $w\in D_k$, $1 \leq i\leq j \leq |w|$, 
where $|w|$ is the length of $w$. We say that ($i$, $j$) is a {\it matched 
pair} of $w$, if $h(w_{i:j})$ is {\it balanced}, i.e., $h(w_{i:j})$ has an 
equal number of $[_1$'s and $]_1$'s and, in any prefix of $h(w_{i:j})$, the 
number of $[_1$'s is greater than or equal to the number of $]_1$'s.
\end{definition}

\begin{lemma}
A string $w\in \{[_1 \hspace{0.1cm},]_1\}^*$ is in $D_1$ if and only
if it is balanced. 
\end{lemma}

Consider the homomorphisms defined as follows

\hspace{-0.4cm}$h_1([_1)\hspace{-0.1cm}=\hspace{-0.1cm}[_1,$ $h_1(]_1)\hspace{-0.1cm}=]_1,$ 
$h_1([_2)\hspace{-0.1cm}=\hspace{-0.1cm}h_1(]_2)\hspace{-0.1cm}=\hspace{-0.1cm}...\hspace{-0.1cm}=
\hspace{-0.1cm}h_1([_k)\hspace{-0.1cm}=\hspace{-0.1cm}h_1(]_k)=\hspace{-0.1cm}\lambda,$

\hspace{-0.4cm}$h_2([_2)\hspace{-0.1cm}=\hspace{-0.1cm}[_1,$ $h_2(]_2)\hspace{-0.1cm}=]_1,$
$h_2([_1)\hspace{-0.1cm}=\hspace{-0.1cm}h_2(]_1)\hspace{-0.1cm}=\hspace{-0.1cm}...\hspace{-0.1cm}=
\hspace{-0.1cm}h_2([_k)\hspace{-0.1cm}=\hspace{-0.1cm}h_2(]_k)\hspace{-0.1cm}=\hspace{-0.1cm}\lambda,$ .\hspace{0.1cm}.\hspace{0.1cm}. , 

\hspace{-0.4cm}$h_k([_k)\hspace{-0.1cm}=\hspace{-0.1cm}[_1,$ $h_k(]_k)\hspace{-0.1cm}=\hspace{0.1cm}]_1,$ 
$h_k([_1)\hspace{-0.1cm}=\hspace{-0.1cm}h_k(]_1)\hspace{-0.1cm}=\hspace{-0.1cm}...\hspace{-0.1cm}=
\hspace{-0.1cm}h_k([_{k-1})\hspace{-0.1cm}=\hspace{-0.1cm}h_k(]_{k-1})\hspace{-0.1cm}=\hspace{-0.1cm}\lambda.$

\begin{lemma}
We have $w \in D_k$, $k\geq 2$, if and only if the following
conditions hold: $i$) $(1$, $|w|)$ is a matched pair, and $ii$) for all 
matched pairs $(i$, $j)$, $h_k(w_{i:j})$  are in $D_1$, $k\geq 1$.
\end{lemma}

\begin{definition}\rm
Let $w\in D_k$, ($i$, $j$) is a {\it nested pair} of $w$ if ($i$, $j$) is
a matched pair, and either $j=i+1$, or ($i+1$, $j-1$) is a matched pair.
\end{definition}

\begin{definition}\rm
 Let $w\in \hspace{-0.1cm}D_k$ and $(i, \hspace{-0.1cm}j)$ a matched pair of $w$. We say
 that $(i, j)$ is {\it reducible} if there exists an integer 
 $j'$, $i\hspace{-0.1cm}<\hspace{-0.1cm}j'\hspace{-0.1cm}<\hspace{-0.1cm}j$, 
such that $(i, j')$ and $(j'\hspace{-0.1cm}+\hspace{-0.1cm}1, j)$ are matched pairs of $w$.
\end{definition}

Consider $w\in D_k$. If ($i$, $j$) is a nested pair of $w$ then ($i$, $j$) 
is an irreducible pair. If ($i$, $j$) is a nested pair of $w$ then 
($i+1$, $j-1$) may be a reducible pair.

\begin{theorem} The trace-language associated with a CFG,
$G=(N_k, T, P_k, S)$ in Dyck normal form, with $|N_k|=2k+1$, is a 
subset of $D_k$. 
\end{theorem}

\begin{proof}
Let $N_k=\{S,[_1,...,[_k,\hspace{0.1cm}]_1,...,\hspace{0.1cm}]_k\}$ be the set of 
nonterminals, $w\in L(G)$, and $D$ a leftmost derivation of $w$. We show that any 
subtree of the derivation tree, read in the depth-first order, by  
ignoring the root and the terminal nodes, corresponds to a matched pair in  
$t_{w,D}$. In particular, $(1, |t_{w,D}|)$ will be a matched pair. Denote by 
${t_{w,D}}_{i:j}$ the substring of $t_{w,D}$ starting at the $i^{th}$ position 
and ending at the $j^{th}$ position of $t_{w,D}$. We show that for all matched 
pairs $(i, j)$, $h_{k'}({t_{w,D}}_{i:j})$ belong to $D_1$, $1\leq k'\leq k$. We 
prove these claims by induction on the height of subtrees.

{\it Basis.} Any subtree of height $n=1$, read in the depth-first order, looks 
like $[_i \hspace{0.1cm}]_i$, $1\leq i\leq k$. Therefore, it satisfies the above 
conditions.

{\it Induction step.} Assume that the claim is true for all subtrees of 
height $h$, $h<n$. We  prove it for $h=n$. Each subtree of height $n$ can have 
one of the following structures. The level $0$ of the subtree is marked by a left 
or right bracket. This bracket will not be considered when we read the subtree. 
Denote by $[_{m}$ the left son of the root. Then the right son is
labeled by $]_{m}$. They are the roots of a left and right subtree,
for which at least one has height $n-1$.

Suppose that both subtrees have the height $1\leq h\leq n-1$. By the induction 
hypothesis, let us suppose further that the left subtree corresponds to the matched 
pair $(i_l,j_l)$, and the right subtree corresponds to the matched pair $(i_r, j_r)$, 
$i_r=j_l+2$, because the position $j_l+1$ is taken by  $]_{m}$. As $h$ is a homomorphism, 
we have $h({t_{w,D}}_{i_l-1:j_r})\hspace{-0.1cm}=
h([_m{t_{w,D}}_{i_l:j_l}]_m{t_{w,D}}_{j_l+2:j_r})\hspace{-0.1cm}=\hspace{-0.1cm}
h([_m)h({t_{w,D}}_{i_l:j_l})h(]_m)h({t_{w,D}}_{j_l+2:j_r})$. Hence, 
$h({t_{w,D}}_{i_l-1:j_r})$ satisfies all conditions in Definition 3.3, and  
thus $(i_l-1,j_r)$ that corresponds to the considered subtree of 
height $n$, is a matched pair. By the induction hypothesis, 
$h_{k'}({t_{w,D}}_{i_l:j_l})$ and 
$h_{k'}({t_{w,D}}_{i_r:j_r})$ are in $D_1$, $1\leq k'\leq k$. Therefore,  
$h_{k'}({t_{w,D}}_{i_l-1:j_r})= h_{k'}([_m)h_{k'}({t_{w,D}}_{i_l:j_l})h_{k'}(]_m)
h_{k'}({t_{w,D}}_{j_l+2:j_r})\in
 \{h_{k'}({t_{w,D}}_{i_l:j_l})h_{k'}({t_{w,D}}_{j_l+2:j_r}),
[_1h_{k'}({t_{w,D}}_{i_l:j_l})]_1h_{k'}({t_{w,D}}_{j_l+2:j_r})\}$
are in $D_1$.

Note that in this case the matched pair $(i_l-1,j_r)$ is reducible
into $(i_l-1,j_l+1)$ and $(j_l+2,j_r)$, where $(i_l-1,j_l+1)$
corresponds to the substring
${t_{w,D}}_{i_l-1:j_l+1}=[_m{t_{w,D}}_{i_l:j_l}]_m$. We refer to 
this structure as the {\it left embedded subtree}, i.e., $(i_l-1,j_l+1)$ 
is a nested pair. A similar reasoning is applied for the case when one 
of the subtrees has the height $0$. 

Analogously, it can be shown that 
the initial tree corresponds to the matched pair $(1, |t_{w,D}|)$, i.e., 
the first condition of Lemma 3.5 holds.

So far, we have proved that each subtree of the derivation tree, and each 
left embedded subtree, corresponds to a matched pair
$(i,j)$ and $(i_l,j_l)$, such that $h_{k'}({t_{w,D}}_{i:j})$ and
$h_{k'}([_m{t_{w,D}}_{i_l:j_l}]_m)$, $1\leq k'\leq k$, are in $D_1$.

Next we show that all matched pairs from $t_{w,D}$ correspond only
to subtrees, or left embedded subtrees, from the derivation tree. To derive 
a contradiction, let us suppose that there exists a matched pair
$(i, j)$ in $t_{w,D}$, that does not correspond to any subtree, or
left embedded subtree, of the derivation tree read in the depth-first order. 

Since $(i, j)$ does not correspond to any subtree, or left embedded subtree, there 
exist two adjacent subtrees $\theta_1$ (a left embedded subtree) and  $\theta_2$ 
(a right subtree) such that $(i, j)$ is composed of two adjacent ``subparts'' of $\theta_1$ 
and $\theta_2$. In terms of matched pairs, if $\theta_1$ corresponds to the matched pair 
$(i_1, j_1)$ and $\theta_2$ corresponds to the matched pair $(i_2, j_2)$, such that $i_2=j_1+2$, 
then there exists a suffix $s_{i_1-1:j_1+1}$ of ${t_{w,D}}_{i_1-1:j_1+1}$, and a prefix 
$p_{i_2:j_2}$ of ${t_{w,D}}_{i_2:j_2}$, such that ${t_{w,D}}_{i:j}=s_{i_1-1:j_1+1}p_{i_2:j_2}$. 
Furthermore, without loss of generality, we assume that $(i_1, j_1)$ and $(i_2, j_2)$ are nested 
pairs. Otherwise, the matched pair $(i, j)$ can be ``narrowed'' until $\theta_1$ and $\theta_2$
are characterized by two nested pairs. If $(i_1, j_1)$ is a nested pair, then so is $(i_1-1, j_1+1)$. 
 
Since $s_{i_1-1:j_1+1}$ is a suffix of ${t_{w,D}}_{i_1-1:j_1+1}$ and $(i_1-1, j_1+1)$ is a matched 
pair, with respect to Definition 3.3, the number of $]_1$'s in $h(s_{i_1-1:j_1+1})$ is greater than 
or equal to the number of $[_1$'s in $h(s_{i_1-1:j_1+1})$. On the other hand,  
$s_{i_1-1:j_1+1}$ is also a prefix of ${t_{w,D}}_{i:j}$, because $(i, j)$ 
is a matched pair, by our hypothesis. Therefore, the number of $[_1$'s in 
$h(s_{i_1-1:j_1+1})$ is greater than or equal to the number of $]_1$'s in $h(s_{i_1-1:j_1+1})$

Hence, the only possibility for $s_{i_1-1:j_1+1}$ to be a suffix for ${t_{w,D}}_{i_1-1:j_1+1}$ and 
a prefix for ${t_{w,D}}_{i:j}$ is the equality between the number of $[_1$'s and $]_1$'s 
in $h(s_{i_1-1:j_1+1})$. This property holds if and only if $s_{i_1-1:j_1+1}$
corresponds to a matched pair in ${t_{w,D}}_{i_1-1:j_1+1}$, i.e., if $i_s$ and $j_s$ are
the start and the end positions of $s_{i_1-1:j_1+1}$ in ${t_{w,D}}_{i_1-1:j_1+1}$, then
$(i_s, j_s)$ is a matched pair. Thus,  $(i_1-1, j_1+1)$ is a reducible pair
into $(i_1-1, i_s-1)$ and $(i_s, j_s)$, where $j_s=j_1+1$. We have reached a 
contradiction, i.e., $(i_1-1, j_1+1)$ is reducible. 

Therefore, the matched pairs in $t_{w,D}$ correspond to subtrees,
or left embedded subtrees, in the derivation tree. For these matched
pairs we have already proved that they satisfy Lemma 3.5. Accordingly,
$t_{w,D} \in D_k$, and consequently the trace-language associated
with $G$ is a subset of $D_k$.  
\end{proof}

\begin{theorem}
Given a CFG  $G$ there exist an integer $K$, a homomorphism $\varphi$, and 
a subset $D'_K$ of the Dyck language $D_K$, such that $L(G)=\varphi(D'_K)$.  
\end{theorem}
 
\begin{proof}
Let $G$ be a CFG and $G_k=(N_k, T, P_k, S)$ be the Dyck normal form  of
$G$, such that $N_k=\{S, [_1,\hspace{0.1cm}[_2,...,[_k,\hspace{0.1cm}]_1,
\hspace{0.1cm}]_2,...,\hspace{0.1cm}]_k \}$. Let $\L_k$ be the
trace-language associated with $G_k$. Consider $\{t_{k+1},...,t_{k+p}\}$ the 
ordered subset of $T$, such that $S\rightarrow t_{k+i} \in P$, $1\leq i\leq p$.  
We define $N_{k+p}=N_k \cup \{[_{t_{k+1}},..., [_{t_{k+p}}, ]_{t_{k+1}},... 
]_{t_{k+p}}\}$, and $P_{k+p}=P_k \cup\{ S\rightarrow [_{t_{k+i}}]_{t_{k+i}}, 
[_{t_{k+i}} \rightarrow t_{k+i}, ]_{t_{k+i}} \rightarrow  \lambda|  S\rightarrow t_{k+i} 
\in P, 1\leq i \leq p\}$. Certainly, the new grammar $G_{k+p}=(N_{k+p}, T, 
P_{k+p}, S)$ generates the same language as $G_k$. 

Let $\varphi: (N_{k+p}-\{S\})^+\rightarrow T^*$ be the homomorphism
defined by $\varphi(N)=\lambda$, for each rule of the form
$N\rightarrow XY$, $N, X, Y\in N_k-\{S\}$, and $\varphi(N)=t$, for
each rule of the form $N\rightarrow t$, $N\in N_k-\{S\}$, and $t\in
T$, $\varphi([_{k+i})=t_{k+i}$ ($\varphi([_{k+i})=\lambda$, if $t_{k+i}=\lambda$), 
and $\varphi(]_{k+i})=\lambda$, for each $1\leq i \leq p$. It is evident that 
$L=\varphi(D'_K)$, where $K=k+p$, $D'_K=\L_k\cup{L_p}$, and 
$L_p=\{[_{t_{k+1}} \hspace{0.1cm}]_{t_{k+1}},..., [_{t_{k+p}}\hspace{0.1cm}]_{t_{k+p}}\}$.
\end{proof}

\section{Even Linear Languages are Contained in $\cal A$$\cal C$$^1$ - An Alternative Proof}

Let $G_k=(N_k, T, P_k, S)$ be an arbitrary CFG in Dyck normal form, with 
$N_k=\{S, [_1,...,[_k,\hspace{0.1cm}]_1,$\\$...,]_k \}$. Let 
$\varphi: (N_k-\{S\})^+ \rightarrow T^+$ be a variant of the homomorphism
introduced in the proof of Theorem 3.9, defined by  $\varphi(N)=\lambda$, 
for each rule of the form $N\rightarrow XY$, $N, X, Y\in N_k$, $N\neq S$, and 
$\varphi(N)=t$, for each rule of the form $N\rightarrow t$, $N\in N_k$, $N\neq S$, 
$t\in T$. Next we divide $N_k$ into three sets $N^{(1)}$, $N^{(2)}$, $N^{(3)}$ as follows 
\vspace{-0.1cm}
\begin{enumerate}
	\item $[_i$ and $]_i$ belong to $N^{(1)}$ iff $\varphi ([_i)=t$ and $\varphi (]_i)=t'$, $t, t' \in T$, 
	\item $[_i$ and $]_i$ belong to $N^{(2)}$ iff $\varphi ([_i)=t$ and $\varphi(]_i)=\lambda$, or 
	   $\varphi ([_i)=\lambda$ and $\varphi (]_i)=t$, $t\in T$, 
	\item $[_i$, $]_i\in N^{(3)}$ iff $\varphi ([_i)=\lambda$ and $\varphi (]_i)=\lambda$. 
\end{enumerate}

Certainly, $N_k-\{S\}=N^{(1)} \cup N^{(2)}\cup N^{(3)}$ and $N^{(1)}\cap N^{(2)}\cap N^{(3)}=\emptyset$. 
The set $N^{(2)}$ is further divided into $N^{(2)}_l$ and $N^{(2)}_r$. The subset $N^{(2)}_l$ 
contains the pairs ($[_i$, $]_i$) $\in N^{(2)}$ for which $\varphi ([_i)\neq\lambda$, while $N^{(2)}_r$ 
contains the pairs ($[_i$, $]_i$) $\in N^{(2)}$ for which $\varphi (]_i)\neq\lambda$. We have 
$N^{(2)}= N^{(2)}_l \cup N^{(2)}_r$ and $N^{(2)}_l \cap N^{(2)}_r=\emptyset$. In the sequel, for each two 
brackets $[_i$ and $]_i$ belonging to $X$, $X\in \{N^{(1)}, N^{(2)}, N^{(3)}\}$, we use the notation 
($[_i$, $]_i$) $\in X$. 

\begin{definition}
An \textit{even linear grammar} (ELG) is a context-free grammar $G=(N, \Sigma, P, S)$  in which  
each rule in $P$ is a linear rule of the form
\begin{enumerate}
\item $X\rightarrow t_1Yt_2$, $t_1,t_2\in \Sigma^*$, $X,Y\in N$ and $|t_1|=|t_2|$, 
\item $X\rightarrow t$, where $t\in\Sigma^*$, $X,Y\in N$. 
\end{enumerate} 
\end{definition}

A language $L$ is even linear if there exists an ELG that generates $L$. Even linear languages 
(henceforth ELIN)  form a proper subclass of the class of linear languages and properly includes 
the class of regular languages \cite{AP}. 

\begin{lemma}
For each even linear grammar $G$, there exits a grammar $G_k$ in Dyck normal form  with $N^{(3)}=\emptyset$, such that $L(G) = L(G_k)$.  
\end{lemma}

An indexing alternating Turing machine (ATM) is an ATM that is allowed to write any binary number 
on a special tape, called {\it index} tape. This number is seen as an address of a location on the 
input tape. With $i$, written in binary on the index tape, the ATM can read the symbol placed on 
the $i^{th}$ cell of the input tape. Using universal states to relate different branches on the 
computation, an indexing ATM can read an input of length $n$, in $\cal O$($\log n$) time. For the formal 
definition and complexity results on ATMs the reader is referred to \cite{BDG}, \cite{CKS}, and \cite{V}. 

Let $G_k=(N_k, T, P_k, S)$ be an ELG in Dyck normal form, and $L(G_k)$ be the language generated by $G_k$. Next we 
sketch an ATM $\cal A$ that decides whether a word $w=a_1a_2...a_n\in T^*$ belongs to $L(G_k)$. Let us 
consider that $n=2p+1$ (the case when $n$ is an even number is similar). The trace-word $t_{v,D}$ 
associated with a word of odd length, $v\in L(G_k)$, $v=v_1v_2...v_{2p+1}$, is of the form A or B, sketched below\vspace{0.1cm}\\

\textbf{A.}\\ 
{\small \hspace*{0.1cm}
\centerline{$t_{x,D}$=\noindent \begin{tabular}{llllllllllllllll}
$\textbf{\small{\hspace{0.1cm}[}}_{j_1}$&$\hspace{-0.1cm}\textbf{\small{]}}_{j_1}$&$\hspace{-0.1cm}[_{i_1}$&$\hspace{-0.1cm}\textbf{\small{[}}_{j_2}$&
$\hspace{-0.2cm}\textbf{\small{]}}_{j_2}$&$\hspace{-0.1cm}[_{i_2}...$&$\hspace{-0.1cm}[_{i_{p-1}}$&$\hspace{-0.1cm}\textbf{\small{[}}_{j_p}$&$
\textbf{\hspace{-0.1cm}\small{]}}_{j_p}$&$\hspace{-0.1cm}[_{i_p}$&$\hspace{-0.1cm}\textbf{\small{[}}^t_{j_{p+1}}$&
$\hspace{-0.2cm}\textbf{\small{]}}^t_{j_{p+1}}$&$\hspace{-0.3cm}]_{i_p}$&$\hspace{-0.1cm}]_{i_{p-1}}...$&$\hspace{-0.1cm}]_{i_2}$&$\hspace{-0.1cm}]_{i_1}$\\
$\hspace{0.2cm}\downarrow$&$\hspace{-0.1cm}\downarrow$&$\hspace{-0.1cm}\downarrow$&$\hspace{-0.1cm}\downarrow$&$\hspace{-0.2cm}\downarrow$&
$\hspace{-0.1cm}\downarrow$&$\hspace{-0.1cm}\downarrow$&$\hspace{-0.1cm}\downarrow$&$\hspace{-0.1cm}\downarrow$&$\hspace{-0.1cm}\downarrow$&
$\hspace{-0.1cm}\downarrow$&$\hspace{-0.1cm}\downarrow$&$\hspace{-0.3cm}\downarrow$&$\hspace{-0.2cm}\downarrow$&$\hspace{-0.1cm}\downarrow$&
$\hspace{-0.1cm}\downarrow$\\
$\hspace{0.1cm}v_1$&$\hspace{-0.2cm}\lambda$&$\hspace{-0.1cm}\lambda$&$\hspace{-0.1cm}v_2$&$\hspace{-0.1cm}\lambda$&$\hspace{-0.1cm}\lambda...$
&$\hspace{-0.1cm}\lambda$&$\hspace{-0.1cm}v_p$&$\hspace{-0.1cm}\lambda$&$\hspace{-0.1cm}\lambda$&$\hspace{-0.2cm}v_{p+1}$&
$\hspace{-0.2cm}v_{p+2}$&$\hspace{-0.3cm}v_{p+3}$&$\hspace{-0.1cm}v_{p+4}...$&$\hspace{-0.2cm}v_{2p}$&$\hspace{-0.2cm}v_{2p+1}$\\
\end{tabular}}\vspace{0.1cm}}\\

\textbf{B.}\\ 
{\small
\centerline{$t_{x,D}$ =\noindent \begin{tabular}{llllllllllllllll}
$\textbf{\small{\hspace{0.1cm}[}}_{j_1}$&$\hspace{0.1cm}\textbf{\small{]}}_{j_1}$&$\hspace{0.1cm}[_{i_1}$&$\hspace{0.1cm}\textbf{\small{[}}_{j_2}$&
$\hspace{0.1cm}\textbf{\small{]}}_{j_2}$&$\hspace{0.1cm}[_{i_2}...$&$\hspace{0.1cm}[_{i_p}$&
&$\hspace{-0.7cm}\textbf{\small{[}}_{j_p}$&$\textbf{\hspace{-0.8cm}\small{]}}_{j_p}$&
$\hspace{-1.2cm}\textbf{\small{[}}^t_{j_{p+1}}$& 
$\hspace{-1.2cm}\textbf{\small{]}}^t_{j_{p+1}}$&$\hspace{-1.5cm}]_{i_p}\hspace{0.1cm}...$&$\hspace{-1cm}]_{i_2}$&$\hspace{-1.4cm}]_{i_1}$\\
$\hspace{0.1cm}\downarrow$&$\hspace{0.1cm}\downarrow$&$\hspace{0.1cm}\downarrow$&$\hspace{0.1cm}\downarrow$&$\hspace{0.1cm}\downarrow$&
$\hspace{0.1cm}\downarrow$&$\hspace{0.1cm}\downarrow$&$\hspace{-0.2cm}\downarrow$&$\hspace{0.1cm}\downarrow$&\hspace{0.2cm}
$\downarrow$&$\hspace{0.2cm}\downarrow$&$\downarrow$&$\hspace{-0.2cm}\downarrow$&
$\hspace{-0.2cm}\downarrow$\\
$\hspace{0.1cm}v_1$&$\hspace{0.1cm}\lambda$&$\hspace{0.1cm}\lambda$&$\hspace{0.1cm}v_2$&$\hspace{0.1cm}\lambda$&$\hspace{0.1cm}\lambda\hspace{0.1cm}...$
&$\hspace{0.1cm}\lambda$&$\hspace{-0.2cm}v_p$&$\hspace{0.2cm}\lambda$&$\hspace{0.2cm}v_{p+1}$&
$\hspace{0.1cm}v_{p+2}$&$\hspace{-0.2cm}v_{p+3}\hspace{0.1cm}...$&$\hspace{-0.2cm}v_{2p}$&$\hspace{-0.1cm}v_{2p+1}$\\
\end{tabular}}\vspace{0.1cm}}\\
where ($[_{j_k}$, $]_{j_k}$) $\in N_l^{(2)}$, ($[_{i_k}$, $]_{i_k}$) $\in N_r^{(2)}$, $1\leq k \leq p$, 
($[^t_{j_{p+1}}$, $]^t_{j_{p+1}}$) $\in N^{(1)}$, and each down arrow represents the homomorphism $\varphi$.

Let $Q_1$ be the quotient and $R_1$ the remainder of $p$ divided\footnote{By $\left[a\right]$ we denote the largest integer not greater 
than $a$, where $a$ is a real number.} by $\left[\log p\right]$. Dividing $Q_1$ by $\left[\log p\right]$ a new quotient $Q_2$ and remainder 
$R_2$ are obtained. If this ``iterated'' division is performed until the resulted quotient, denoted by $Q_\ell$, can be no longer divided by 
$\left[\log p\right]$, then $p$ is $p=((...((Q_\ell\left[\log p\right]+R_\ell)\left[\log p\right]+R_{\ell-1})\left[\log p\right]+...)\left[\log p\right]+R_2)\hspace{-0.1cm}\left[\log p\right]+R_1$, $1\leq Q_\ell<\hspace{-0.1cm}\log p$, 
$0\leq R_l<\hspace{-0.1cm}\log p$, $l\in\{1,...,\ell\}$,  and $\ell<\log p$. 

Knowing $R_1$, $\cal A$ guesses an $R_1$-tuple of left brackets in $N_l^{(2)}$ and it checks whether brackets in $R_1$ can produce, 
according to the derivation in $G_k$, the first $R_1$ terminal symbols $a_1, a_2,..., a_{R_1}$ in $w$, and whether the right 
pairwise bracket of each left bracket in the $R_1$-tuple produces a right bracket in $N_r^{(2)}$ that at its turn produces the terminal 
symbol opposed, in $w$, to the terminal symbol produced by the left bracket in $R_1$. $\cal A$ guesses a $\left[\log p\right]$-tuple 
of left brackets in $N_l^{(2)}$ supposed to produce those terminals placed at the $\left[\log p\right]$ cutting points in $w$ obtained by 
dividing $[R_1+1...p]$ into $\left[\log p\right]$ intervals of length $Q_1$. 

In parallel, $\cal A$ checks whether there exists an $R_2$-tuple of left brackets in $N_l^{(2)}$ that generate the first $R_2$ terminal 
strings in each $Q_1$-interval in $a_1a_2...a_p$ according to the derivation in $G_k$, and whether the right pairwise bracket of each left 
bracket in the $R_2$-tuple produces a right bracket in $N_r^{(2)}$ that at its turn produces the terminal symbol opposed, in $w$, to the 
terminal symbol produced by the left bracket in $R_2$. In parallel for each $Q_1$-interval, $\cal A$ 
guesses another $\left[\log p\right]$-tuple of left brackets in $N_l^{(2)}$ supposed to produce those terminals placed at the $\left[\log p\right]$ 
cutting points in $w$, obtained by dividing each interval of length $Q_1-R_2$ into $\left[\log p\right]$ intervals of length $Q_2$. This procedure 
is repeated until intervals of length $Q_\ell <\log p$ are obtained. At this point, $\cal A$ checks whether the substrings of $w$ corresponding 
to the $Q_\ell$-intervals in $a_1a_2...a_p$ produced by left brackets in $N_l^{(2)}$, and their corresponding opposite substrings in 
$a_{p+1}a_{p+2}...a_{2p}a_{2p+1}$, produced by right brackets in $N_r^{(2)}$,  are valid according to the derivation in $G_k$. If correct 
$\left[\log p\right]$-tuples of left brackets in $N_l^{(2)}$ can be found for all intervals and all cutting points, then $w\in L(G_k)$.  
We have        

\begin{theorem}
Each language $L\in ELIN$ can be accepted by an indexing alternating Turing machine in 
$\cal O$($\log n$) space, $\cal O$($\log^2n$) time, and a logarithmic number of alternations. 
\end{theorem}

\begin{proof}
Let $G_k=(N_k, T, P_k, S)$ be an arbitrary ELG in Dyck normal form, and $N_k=\{S, [_1,...,[_k, \hspace{0.1cm}]_1,...,\hspace{0.1cm}]_k \}$. 
Let $\cal A$ be an indexing ATM composed of an input tape that stores an input word $w\in T^*$ of length $n$, $w=a_1a_2...a_n$,  
an index tape, and a working tape composed of two tracks initially empty.\vspace{0.1cm}\\ 
\textbf{Level 1}(\textit{Existential}) In an existential state $\cal A$ guesses 
the length of the input string, and verifies the correctness of this guess. 
This is performed by writing, for each existential branch, on the index tape the 
binary value of $n$, and checking afterwards whether the $n^{th}$ cell of the 
input tape contains a terminal symbol and the $(n+1)^{th}$ cell contains no symbol. 
The correct value of $n$ is recorded in binary on the first track of the working tape of $\cal A$.\vspace{0.1cm}\\ 
\textbf{Level 2} (\textit{Existential}) Consider the quotient $Q_1$ and  the remainder $R_1$ of the division of $p$ by 
$\left[\log p\right]$, $0\leq R_1\hspace{-0.1cm}<\left[\log p\right]$. 
$\cal A$ spawns $\cal O$($c^{\log p}$) existential branches, 
each branch holding an $R_1$-tuple of brackets $\Re_{R_1}=([_{j_1},...,[_{j_k},..., [_{R_1})$, where each $[_{j_k}$, $1\leq k \leq R_1$, 
is a right bracket in $N_l^{(2)}$ supposed to produce\footnote{The $R_1$-tuple makes a guess of the first $R_1$ right brackets 
in $t_{w,D}$, such that the homeomorphic image of $[_{j_1}...[_{j_k}...[_{j_{R_1}}$ through $\varphi$ to be the prefix $a_1...a_k...a_{R_1}$. 
Certainly there exist $\cal O$($c^{\log p}$) such guesses, where $c=|N_r^{(2)}|$.} $a_k$. Then $\cal A$ universally checks (Level 3) whether 
each bracket in $\Re_{R_1}$ is a correct guess.\vspace{0.1cm}\\
\textbf{Level 3} (\textit{Universal}) $\cal A$ universally branches all substring $a_ka_{k+1}$, $1\leq k \leq R_1-1$, and checks 
whether there exist the rules $[_{j_k}\rightarrow a_k$, $[_{j_{k+1}}\rightarrow a_{k+1}$, $]_{j_k}\rightarrow [_{i_k}\hspace{0.1cm}]_{i_k}$, 
$[_{i_k}\rightarrow [_{j_{k+1}}\hspace{0.1cm}]_{j_{k+1}}$, and $]_{i_k}\rightarrow a_{2p-k+2}$ in $P$. If these conditions hold for each 
universal branch, then $\Re_{R_1}$ is a correct guess and the existential branch from Level 2, that spawned $\Re_{R_1}$, will be labeled 
by $1$.\vspace{0.1cm}\\   
\textbf{Level 4} (\textit{Existential}) Let $Q_2$ be the quotient and $R_2$ the remainder of $Q_1$ divided by $\left[\log p\right]$, 
$0\leq R_2<\left[\log p\right]$. $\cal A$ spawns $\cal O$($c^{\log p}$) existential branches, each branch holding a 
$2\left[\log n\right]$-tuple $\Re^c_{R_2}\hspace{-0.1cm}=\hspace{-0.1cm}([_{j_{R_1}}, [_{j_{R_1+R_2}}, [_{j_{R_1+Q_1}}, [_{j_{R_1+Q_1+R_2}}, 
...,[_{j_{R_1+(\left[\log n\right]-1)Q_1}}, [_{j_{R_1+(\left[\log n\right]-1)Q_1+R_2}})$,  where $[_{j_{R_1}}$ is the left bracket in 
$N_l^{(2)}$, belonging to the tuple $\Re_{R_1}$ found correct at Level 3, and each $[_{j_{R_1+i_1Q_1}}$ and $[_{j_{R_1+i_1Q_1+R_2}}$ is a bracket 
in $N_l^{(2)}$ supposed to produce the terminal symbols placed at the ``cutting points''  $R_1+i_1Q_1$ and $R_1+i_1Q_1+R_2$, respectively, i.e., $a_{j_{R_1+i_1Q_1}}$ and $a_{j_{R_1+i_1Q_1+R_2}}$, where $0\leq i_1 \leq \left[\log p\right]-1$. Because $\Re_{R_1}$ is no more needed, the 
space used by $\cal A$ to record $\Re_{R_1}$ is allocated now to record $\Re^c_{R_2}$.\vspace{0.1cm}\\
\textbf{Level 5} (\textit{Universal}) On each existential branch from Level 4, $\cal A$ spawns $\left[\log p\right]$ universal 
processes $\wp^{(Q_1)}_{i_1}$, $0\leq i_1 \leq \left[\log p\right]-1$. Each process $\wp^{(Q_1)}_{i_1}$ takes the interval 
$[R_1+i_1Q_1...R_1+i_1Q_1+R_2]$, and it checks whether the brackets $[_{j_{R_1+i_1Q_1}}$ and $[_{j_{R_1+i_1Q_1+R_2}}$, 
 $1\leq i_1 \leq \left[\log n\right]-1$, are indeed correct guesses for $a_{j_{R_1+i_1Q_1}}$ and $a_{j_{R_1+i_1Q_1+R_2}}$, respectively.  
Besides $[_{j_{R_1+i_1Q_1}}$ and $[_{j_{R_1+i_1Q_1+R_2}}$, each $\wp^{(Q_1)}_{i_1}$ also keeps, from the previous level, the bracket 
$[_{j_{R_1+(i_1+1)Q_1}}$. In this way each bracket $[_{j_{R_1+i_1Q_1}}$, $1\leq i_1 \leq \left[\log p\right]-1$, is redirected to only one 
process, i.e., $\wp^{(Q_1)}_{i_1-1}$. \vspace{0.1cm}\\
\textbf{Level 6} (\textit{Existential}) For each universal process $\wp^{(Q_1)}_{i_1}$, $0\leq i_1 \leq \left[\log n\right]-1$, $\cal A$ 
spawns $\cal O$($c^{\log p}$) existential branches, each branch holding an $(R_2+1)$-tuple of brackets 
$\Re_{R_2}\hspace{-0.1cm}=([_{j_{R_1+i_1Q_1}}, [_{R_1+i_1Q_1+1},...,[_{R_1+i_1Q_1+R_2-1}, [_{R_1+i_1Q_1+R_2})$. 
Then $\cal A$ checks whether all brackets composing $\Re_{R_2}$ are correct guesses. This can be done, for each process 
$\wp^{(Q_1)}_{i_1}$ through Level 7 (similar to Level 3) as follows. \vspace{0.1cm}\\
\textbf{Level 7} (\textit{Universal}) For each existential branch spawned at Level 6, $\cal A$ universally branches all substring 
$a_ka_{k+1}$, $R_1+i_1Q_1\leq k \leq R_1+i_1Q_1+R_2-1$, and checks whether there exist the rules $[_{j_k}\rightarrow a_k$, 
$[_{j_{k+1}}\rightarrow a_{k+1}$, $]_{j_k}\rightarrow [_{i_k}\hspace{0.1cm}]_{i_k}$, $[_{i_k}\rightarrow [_{j_{k+1}}\hspace{0.1cm}]_{j_{k+1}}$, 
$]_{i_k}\rightarrow a_{2p-k+2}$ in $P$. If these conditions hold for each universal branch, then the corresponding tuple $\Re_{R_2}$ is 
said {\it partially correct} and the existential branch from Level 6, that spawned $\Re_{R_2}$, will be labeled by $\diamond$.\vspace{0.1cm}\\
Note that, at this moment we cannot decide whether the existential branch that spawned $\Re_{R_2}$ can be labeled by $1$, even if all 
universal branches at this level return the true value. 

This is because we do not know yet, whether the bracket $[_{j_{R_1+i_1Q_1}}$ 
is a correct guess, i.e., whether rules of the form $]_{j_{R_1+i_1Q_1-1}}\rightarrow [_{i_{R_1+i_1Q_1-1}}\hspace{0.1cm}]_{i_{R_1+i_1Q_1-1}}$,
$[_{i_{R_1+i_1Q_1-1}}\rightarrow [_{j_{R_1+i_1Q_1}}\hspace{0.1cm}]_{j_{R_1+i_1Q_1}}$, $[_{j_{R_1+i_1Q_1-1}}\rightarrow a_{R_1+i_1Q_1-1}$,  
$]_{i_{R_1+i_1Q_1-1}}\rightarrow a_{2p-(R_1+i_1Q_1-1)+2}$ exist in $P$, since the bracket $[_{j_{R_1+i_1Q_1-1}}$ has not been guessed\footnote{The
left bracket $[_{j_{R_1+i_1Q_1-1}}$ will be guessed at the last level of the computation tree of $\cal A$, when all the remainders of the ``iterated'' 
division of $p$ by $\left[\log p\right]$ will be spent, and when $[_{j_{R_1+i_1Q_1-1}}$ will be the last bracket occurring in the suffix of length 
$Q_\ell$ of the subword $a_{R_1+(i_1-1)Q_1}...a_{R_1+i_1Q_1-1}$ of $w$.} yet. The logical value of each tuple $\Re_{R_2}$ will be decided at the end 
of computation, when it will be known whether bracket $[_{j_{R_1+i_1Q_1-1}}$ is a correct guess, with respect to the other left brackets guessed for 
the subword $a_{R_1+(i_1-1)Q_1}...a_{R_1+i_1Q_1-1}$. If $\Re_{R_2}$ is not partially correct, it is labeled by $0$. 

Suppose that we have run the algorithm up to the $(\ell-1)^{th}$ ``iterated'' division of $p$ by $\left[\log p\right]$, i.e., we know the 
quotient $Q_{\ell-1}$ and the remainder $R_{\ell-1}$ of $Q_{\ell-2}$ divided by $\left[\log p\right]$. \vspace{0.1cm}\\
\textbf{Level 4$(\ell-1)$} (\textit{Existential}) Let $Q_\ell$ be the quotient and $R_\ell$ the remainder of $Q_{\ell-1}$ divided by 
$\left[\log p\right]$, $0\leq Q_\ell, R_\ell<\left[\log p\right]$. Since $Q_{\ell-2}$, $R_{\ell-2}$ and $R_{\ell-1}$ are no more needed, 
the space used to record them is now used to record $Q_\ell$ and $R_\ell$ in binary, still keeping $Q_{\ell-1}$. Denote by 
$x_{i_{\ell-2}}=\sum_{l=1}^{\ell-1}R_l+\sum_{l=1}^{\ell-2}i_lQ_l$. For each existential branch labeled by $\diamond$ at Level $4\ell-6$, 
$\cal A$ spawns $\cal O$($c^{\log p}$) existential branches, each branch holding a $2\left[\log p\right]$-tuple 
$\Re^c_{R_\ell}=([_{j_{x_{i_{\ell-2}}}}, [_{j_{x_{i_{\ell-2}}+\hspace{0.1cm}R_\ell}}, [_{j_{x_{i_{\ell-2}}+\hspace{0.1cm}Q_{\ell-1}}}, 
[_{j_{x_{i_{\ell-2}}+Q_{\ell-1}+R_\ell}},...,[_{j_{x_{i_{\ell-2}}+(\left[\log n\right]-1)Q_{\ell-1}}}, [_{j_{x_{i_{\ell-2}}+(\left[\log n\right]-1)Q_{\ell-1}+R_\ell}})$, 
where $[_{j_{x_{i_{\ell-2}}}}$ is the left bracket belonging to tuple $\Re_{R_{\ell-1}}$ found correct at Level 4$\ell-5$. Because 
$\Re_{R_{\ell-1}}$ is no more needed the space used to record $\Re_{R_{\ell-1}}$ is allocated now to record 
$\Re^c_{R_\ell}$. Then $\cal A$ proceeds with Level 4$\ell-3$, similar to Levels 5,..., 4$\ell-7$. \vspace{0.1cm}\\
\textbf{Level 4$\ell-3$} (\textit{Universal}) On each existential branch spawned at Level 4$(\ell-1)$, $\cal A$ spawns $\left[\log p\right]$ 
universal processes $\wp^{(Q_{\ell-1})}_{i_{\ell-1}}$\hspace{-0.1cm}, $0\hspace{-0.1cm}\leq \hspace{-0.1cm}i_{\ell-1}\hspace{-0.1cm}\leq \hspace{-0.1cm}\left[\log p\right]\hspace{-0.1cm}-\hspace{-0.1cm}1$. Denote by $x_{i_{\ell-1}}\hspace{-0.1cm}=\hspace{-0.1cm}x_{i_{\ell-2}}
\hspace{-0.1cm}+i_{\ell-1}Q_{\ell-1}$,  $0\hspace{-0.1cm}\leq \hspace{-0.1cm}i_{\ell-1}\hspace{-0.1cm}\leq \hspace{-0.1cm}\left[\log p\right]\hspace{-0.1cm}-1$. Each  $\wp^{(Q_{\ell-1})}_{i_{\ell-1}}$ takes the interval 
$[x_{i_{\ell-1}}...x_{i_{\ell-1}}+R_\ell]$, and checks whether the left brackets (spawned at Level 4$(\ell-1)$) 
$[_{j_{x_{i_{\ell-1}}}}$ and $[_{j_{x_{i_{\ell-1}}+R_\ell}}$  are correct guesses. Besides $[_{j_{x_{i_{\ell-1}}}}$ and $[_{j_{x_{i_{\ell-1}}+R_\ell}}$, 
each $\wp^{(Q_{\ell-1})}_{i_{\ell-1}}$, also keeps from the previous level the left bracket $[_{j_{x_{i_{\ell-2}}+(i_{\ell-1}+1)Q_{\ell-1}}}$. 
Then $\cal A$ continues with Level 4$\ell-2$, similar to Levels 6,..., 4$\ell-6$. \vspace{0.1cm}\\
\textbf{Level 4$\ell-2$} (\textit{Existential}) For each universal process $\wp^{(Q_{\ell-1})}_{i_{\ell-1}}$, 
$0\leq {i_{\ell-1}} \leq \left[\log p\right]-1$, $\cal A$ spawns $\cal O$($c^{\log p}$) existential branches, each branch 
holding an $(R_\ell+1)$-tuple of brackets $\Re_{R_\ell}\hspace{-0.1cm}=\hspace{-0.1cm}([_{j_{x_{i_{\ell-1}}}},
[_{j_{x_{i_{\ell-1}}+1}},...,[_{j_{x_{i_{\ell-1}}+R_\ell-1}}, [_{j_{x_{i_{\ell-1}}+R_\ell}})$. Then $\cal A$ checks whether 
all brackets composing $\Re_{R_\ell}$ are correct.  This can be done, for each process $\wp^{(Q_{\ell-1})}_{i_{\ell-1}}$, through 
Level 4$\ell-1$ similar to Levels 3, 7,..., 4$\ell-5$.\vspace{-0.3cm}\\

At this point the only substrings of $w$ left unchecked are those substrings that corresponds to the intervals of the form 
$I_{Q_{\ell-1}}$ = $[\sum_{l=1}^{\ell-1}R_l+\sum_{l=1}^{\ell-2}i_lQ_l+i_{\ell-1}Q_{\ell-1}+R_\ell... 
\sum_{l=1}^{\ell-1}R_l+\sum_{l=1}^{\ell-2}i_lQ_l +(i_{\ell-1}+1)Q_{\ell-1}]$,  $0\leq i_l\leq \left[\log n\right]-1$, 
$1\leq l \leq \ell-1$, and besides the cutting points $P_\ell^u=\sum_{l=1}^uR_l+\sum_{l=1}^{u-1}i_lQ_l+(i_u+1)Q_u$,  
$1\leq u\leq \ell-1$. On each interval of type $I_{Q_{\ell-1}}$, $\cal A$ proceeds with Level 4$\ell$. \vspace{0.1cm}\\
\textbf{Level $4\ell$} (\textit{Existential}) Each interval $I_{Q_{\ell-1}}$ can be divided into 
$\left[\log p\right]$ subintervals of length $1\leq Q_\ell<\hspace{-0.1cm}\left[\log p\right]$. Hence, $\cal A$ spawns 
$\cal O$($c^{\log p}$) existential branches each of which holds a $\left[\log p\right]$-tuple  $\Re^c_{Q_\ell}\hspace{-0.1cm}=\hspace{-0.1cm}([_{j_{x_{i_{\ell-1}}\hspace{-0.1cm}+R_\ell}}, [_{j_{x_{i_{\ell-1}}+R_\ell+Q_\ell}},...,
[_{j_{x_{i_{\ell-1}}\hspace{-0.1cm}+R_\ell+(\left[\log n\right]-1)Q_\ell}})$, where $[_{j_{x_{i_{\ell-1}}+R_\ell}}$ is the left 
bracket found valid at Level $4\ell-1$.\vspace{0.1cm}\\
\textbf{Level $4\ell+1$} (\textit{Universal}) For each existential branch spawned at Level $4\ell$, $\cal A$ spawns 
$\left[\log p\right]$ universal processes $\wp^{(Q_\ell)}_{i_\ell}$, $0\leq i_\ell\leq \left[\log p\right]-1$. Each 
$\wp^{(Q_\ell)}_{i_\ell}$ takes an interval of length $Q_\ell$ of the form $[x_{i_\ell}...x_{i_\ell+1}]$, where  $x_{i_\ell}\hspace{-0.1cm}=\hspace{-0.1cm}\sum_{l=1}^{\ell}R_l+\hspace{-0.1cm}\sum_{l=1}^{\ell-1}i_lQ_l+i_{\ell}Q_{\ell}$, 
$0\hspace{-0.1cm}\leq i_\ell\hspace{-0.1cm}\leq \left[\log n\right]-1$. On each interval $[x_{i_\ell}...x_{i_\ell+1}]$,  $\cal A$ 
checks (Level $4\ell+2$) whether $a_{x_{i_\ell}}...$ $a_{x_{i_\ell+1}}$ and its opposite substring in $w$, 
$a_{2p-x_{i_\ell+1}+2}...$ $a_{2p-x_{i_\ell}+2}$ can be generated by $G_k$.\\
\textbf{Level $4\ell+2$} (\textit{Existential}) For each $\wp^{(Q_{\ell})}_{i_{\ell}}$, $0\leq {i_{\ell}} \leq \left[\log p\right]-1$, 
$\cal A$ spawns $\cal O$($c^{\log p}$) existential branches, each branch holding an $(Q_\ell+1)$-tuple of left brackets   
$\Re_{Q_\ell}\hspace{-0.1cm}=\hspace{-0.1cm}([_{j_{x_{i_\ell}}}\hspace{-0.1cm}, [_{j_{x_{i_\ell}+1}},...,[_{j_{x_{i_\ell}+Q_\ell-1}},$\\$
[_{j_{x_{i_\ell+1}}})$. In each $\Re_{Q_\ell}$ brackets $[_{j_{x_{i_\ell}}}$ and $[_{j_{x_{i_\ell+1}}}$ have been guessed 
at Level 4$\ell$. They are brackets supposed to produce the terminal symbols placed in cutting points, i.e., edges of intervals 
of length $\left[\log p\right]$. These  cutting points are also overlapping points of two consecutive intervals of type  
$[x_{i_\ell}...x_{i_\ell+1}]$. Hence, each bracket $[_{j_{x_{i_\ell}}}$  is checked two times. First, if $[_{j_{x_{i_\ell}}}$ can 
be followed in the trace-word by $[_{j_{x_{i_\ell+1}}}$, producing $a_{x_{i_\ell}}$, and there exists a right bracket $]_{i_{x_{i_\ell}}}$ 
producing $a_{2p-x_{i_\ell}+2}$ (this is checked by $\wp^{(Q_{\ell})}_{i_{\ell}}$). Second, if $[_{j_{x_{i_\ell}}}$ can be preceded 
in the trace-word by $[_{j_{x_{i_{\ell}-1}}}$, producing $a_{x_{i_\ell}}$, and there exists a right bracket $]_{i_{x_{i_{\ell}-1}}}$ 
producing $a_{2p-(x_{i_{\ell}-1})+2}$ (this is checked by process $\wp^{(Q_{\ell})}_{i_{\ell}-1}$).\vspace{-0.4cm}\\  

As all intervals $[x_{i_{\ell}}...x_{i_{\ell}+1}]$ are universally checked, $\Re^c_{Q_\ell}$ is correct (and the 
corresponding existential branch is labeled by $1$) if all brackets in  $\Re_{Q_\ell}$ and all $[_{j_{x_{i_{\ell}}}}$, 
$0\leq i_\ell\leq \left[\log p\right]-1$, are correct guesses. To check whether all left brackets in $\Re_{Q_\ell}$ are correct, 
$\cal A$ follows the same procedure, described at Levels 3,..., $4\ell-1$. 

It can be proved that each cutting point $P_\ell^u=\sum_{l=1}^uR_l+\sum_{l=1}^{u-1}i_lQ_l+(i_u+1)Q_u$ 
yielded at Level $4u$ by $\Re^c_{R_{u+1}}$, $1\leq u \leq \ell-1$, is in fact the right edge of an interval of the form  
$[\sum_{l=1}^{\ell}R_l+\sum_{l=1}^{\ell-1}i_lQ_l+i_{\ell}Q_{\ell}...\sum_{l=1}^{\ell}R_l+\sum_{l=1}^{\ell-1}i_lQ_l+(i_\ell+1)Q_\ell]
=[x_{i_\ell}...x_{i_\ell+1}]$ for which $0\leq i_l\leq \left[\log n\right]-1$, $1\leq l \leq \ell-1$, $i_\ell=\left[\log \right]-1$.
Hence, the decision on the correctness of each left bracket  $[_{j_{P_\ell^u}}$ will be actually taken by a process of type 
$\wp^{(Q_\ell)}_{\left[\log p\right]-1}$. Since the validity of each cutting point is decided by a process 
of type $\wp^{(Q_\ell)}_{\left[\log n\right]-1}$, the logical value returned by this process is "propagated" up to the level of the 
computation tree that has spawned the corresponding cutting point, and thus each $\diamond$ symbol receives a logical value. The input 
is accepted, if going up in the computation tree, with all $\diamond$'s changed into logical values, the root of the tree is labeled by $1$. 

The tuples $\Re_{R_\hbar}$, $\Re^c_{R_\hbar}$, $\Re_{Q_\ell}$, $\Re^c_{Q_\ell}$,  $1\leq \hbar\leq \ell$, are stored by using $\cal O$($\log n$) 
space, on the second track of the working tape of $\cal A$. It is easy to observe that $\cal A$ has $\cal O$($\log n$) levels. Since at each level 
$\cal A$ spawns $\cal O$($c^{\log n}$) existential branches, where $c$ is a constant (each level being thus convertible, by a divide and conquer 
procedure, into a binary tree with $\cal O$($\log n$) levels), and at each Level 4$\hbar$, $1\leq \hbar\leq \ell$,  $\cal A$ performs a division 
(of small numbers) operation, which requires at most $\cal O$($\log n$) time and space \cite{DHK}, $\cal A$ will perform the whole computation in 
$\cal O$($\log^2 n$) time and $\cal O$($\log n$) space. 
\end{proof} 

\begin{corollary}
$ELIN \subseteq$ $\cal A$$\cal C$$^1$. 
\end{corollary}

\begin{proof}
The computation tree of the ATM described in the proof of Theorem 4.3 has a logarithmic number of alternations. This number is not affected by 
the complexity of the division operation, or of the conversion procedure of a positive integer into its binary representation, since all these 
problems can be computable in $\cal A$$\cal C$$^0$ \cite{DHK}. According to the characterization of the $\cal A$$\cal C$ classes in 
terms of ATM resources \cite{C}, \cite{V}, we have $U_L$-uniform $\cal A$$\cal C$$^i = AALT - SPACE(\log^in,\log n)$, for all $i\geq 1$, and 
$\cal A$$\cal C$$^0 = AALT - TIME(CON,\log n)$, where  $AALT - SPACE(\log^in,\log n)$ is the class of problems computable by an ATM by using 
$\cal O$($\log n$) space and $\cal O$($\log^i n$) number of alternations, while  $AALT - TIME(CON,\log n)$ is the class of problems computable 
by an ATM by using $\cal O$($\log n$) time  and a constant number of alternations.             
\end{proof}

\end{document}